\documentclass[10pt,a4paper]{article}
\newcommand{\onlySHORT}[1]{} \newcommand{\onlyFULL}[1]{#1}

\onlyFULL{\usepackage{fullpage}}
\onlyFULL{\usepackage{authblk}}

\usepackage[utf8]{inputenc}
\usepackage{microtype}

\usepackage{amsmath}
\usepackage{amssymb}
\usepackage{amsthm}
\usepackage{hyperref}
\usepackage{mathtools}

\usepackage{todonotes}

\usepackage{verbatim}

\onlyFULL{
\newtheorem{theorem}{Theorem}
\newtheorem{corollary}[theorem]{Corollary}
\newtheorem{lemma}[theorem]{Lemma}

\newtheorem{definition}[theorem]{Definition}
}  % END of \onlyFULL
\newtheorem{observation}[theorem]{Observation}
\newtheorem{hypothesis}[theorem]{Hypothesis}

\let\leq\leqslant
\let\geq\geqslant
\let\le\leqslant
\let\ge\geqslant

\let\eps\varepsilon

\newcommand{\brac}[1]{{\left(#1\right)}}
\newcommand{\set}[1]{\left\{#1\right\}}

\newcommand{\floor}[1]{{\left\lfloor #1 \right\rfloor}}
\newcommand{\ceil}[1]{{\left\lceil #1 \right\rceil}}
\newcommand{\seq}[1]{\left\langle #1 \right\rangle}

\newcommand{\Oh}[1]{\mathcal{O}\left(#1\right)}
\newcommand{\oh}[1]{o\left(#1\right)}

\newcommand{\occ}{L}
\newcommand{\conc}{\circ}
\newcommand{\poly}{\mathrm{poly}}
\newcommand{\polylog}{\mathrm{polylog}}

\newcommand{\Nat}{\mathbb{N}}
\newcommand{\lcis}{\mathop{\mathrm{lcis}}}
\newcommand{\infl}{\mathop{\mathrm{inflate}}}
\newcommand{\nullv}{\mathop{\mathrm{null}}}
\newcommand{\dotp}{\boldsymbol{\cdot}}

\newcommand{\VG}{\mathrm{VG}}
\newcommand{\CG}{\mathrm{CG}}
\newcommand{\x}{\textsc{x}}
\newcommand{\y}{\textsc{y}}

\makeatletter
\newcommand\ie{i.e\@ifnextchar.{}{.\@}}
\newcommand\etc{etc\@ifnextchar.{}{.\@}}
\newcommand\etal{et~al\@ifnextchar.{}{.\@}}
\makeatother

\title{Tight Conditional Lower Bounds for Longest Common Increasing Subsequence}

\author[1]{Lech Duraj\thanks{Partially supported by Polish National Science Center grant 2016/21/B/ST6/02165.}}
\author[2]{Marvin K\"unnemann}
\author[1]{Adam Polak\thanks{Partially supported by Polish Ministry of Science and Higher Education program {\em Diamentowy Grant}.}}

\affil[1]{
Theoretical Computer Science\\
Faculty of Mathematics and Computer Science\\
Jagiellonian University, Krak\'{o}w, Poland\\
\texttt{\{duraj,polak\}@tcs.uj.edu.pl}
}
\affil[2]{
Max Planck Institute for Informatics, Saarland Informatics Campus, Saarbr\"ucken, Germany\\
\texttt{marvin@mpi-inf.mpg.de}
}

\onlyFULL{\date{}}

\onlySHORT{

\authorrunning{L. Duraj, M. K\"unnemann and A. Polak}

\Copyright{Lech Duraj, Marvin K\"unnemann and Adam Polak}


\EventEditors{Daniel Lokshtanov and Naomi Nishimura}
\EventNoEds{2}
\EventLongTitle{12th International Symposium on Parameterized and Exact Computation (IPEC 2017)}
\EventShortTitle{IPEC 2017}
\EventAcronym{IPEC}
\EventYear{2017}
\EventDate{September 6–-8, 2017}
\EventLocation{Vienna, Austria}
\EventLogo{}
\SeriesVolume{89}
\ArticleNo{15} % “New number” (=<article-no>) goes here!
} % END of \onlySHORT

\begin{document}
\maketitle

\begin{abstract}
We consider the canonical generalization of the well-studied Longest Increasing Subsequence problem to multiple sequences, called $k$-LCIS: Given $k$ integer sequences $X_1,\dots,X_k$ of length at most $n$, the task is to determine the length of the longest common subsequence of $X_1,\dots,X_k$ that is also strictly increasing. Especially for the case of $k=2$ (called LCIS for short), several algorithms have been proposed that require quadratic time in the worst case.

Assuming the Strong Exponential Time Hypothesis (SETH), we prove a tight lower bound, specifically, that no algorithm solves LCIS in (strongly) subquadratic time. Interestingly, the proof makes no use of normalization tricks common to hardness proofs for similar problems such as LCS. We further strengthen this lower bound (1) to rule out $\Oh{ (nL)^{1-\varepsilon}}$ time algorithms for LCIS, where $L$ denotes the solution size, (2) to rule out $\Oh{n^{k-\varepsilon}}$ time algorithms for $k$-LCIS, and (3) to follow already from weaker variants of SETH. We obtain the same conditional lower bounds for the related Longest Common Weakly Increasing Subsequence problem.
\end{abstract}

\section{Introduction}

The longest common subsequence problem (LCS) and its variants are computational primitives with a variety of applications, which includes, e.g., uses as similarity measures for spelling correction~\cite{Morgan70,WagnerF74} or DNA sequence comparison~\cite{NeedlemanW70,AltschulGMML90}, as well as determining the differences of text files as in the UNIX \textsc{diff} utility~\cite{HuntMcI75}. LCS shares characteristics of both an easy and a hard problem: (\emph{Easy}) A simple and elegant dynamic-programming algorithm computes an LCS of two length-$n$ sequences in time $\Oh{n^2}$~\cite{WagnerF74}, and in many practical settings, certain properties of typical input sequences can be exploited to obtain faster, ``tailored'' solutions (e.g., ~\cite{Hirschberg77,HuntS77,ApostolicoG87,Myers86}; see also \cite{BergrothHR00} for a survey). (\emph{Hard}) At the same time, no polynomial improvements over the classical solution are known, thus exact computation may become infeasible for very long general input sequences. The research community has sought for a resolution of the question \emph{``Do subquadratic algorithms for LCS exist?''} already shortly after the formalization of the problem~\cite{ChvatalKK72,AhoHU76}.

Recently, an answer conditional on the Strong Exponential Time Hypothesis (SETH; see Section~\ref{sec:prelim} for a definition) could be obtained: Based on a line of research relating the satisfiability problem to quadratic-time problems~\cite{Williams05,RodittyVW13,Bringmann14,AbboudVWW14} and following a breakthrough result for Edit Distance~\cite{BackursI15}, it has been shown that unless SETH fails, there is no (strongly) subquadratic-time algorithm for LCS~\cite{AbboudBVW15,BringmannK15}. Subsequent work~\cite{AbboudHVWW16} strengthens these lower bounds to hold already under weaker assumptions and even provides surprising consequences of sufficiently strong polylogarithmic improvements.

Due to its popularity and wide range of applications, several variants of LCS have been proposed. This includes the heaviest common subsequence (HCS)~\cite{JacobsonV92}, which introduces weights to the problem, as well as notions that constrain the structure of the solution, such as the longest common increasing subsequence (LCIS)~\cite{YangHC05}, LCSk~\cite{BensonLMNS16}, constrained LCS~\cite{Tsai03,ChinSFHK04, ArslanE05}, restricted LCS~\cite{GotthilfHLL10}, and many other variants (see, e.g.,~\cite{ChenC11,AnnYT14,JiangLMZ04}). Most of these variants are (at least loosely) motivated by biological sequence comparison tasks. To the best of our knowledge, in the above list, LCIS is the only LCS variant for which (1) the best known algorithms run in quadratic time in the worst case and (2) its definition does not include LCS as a special case (for such generalizations of LCS, the quadratic-time SETH hardness of LCS~\cite{AbboudBVW15,BringmannK15} would transfer immediately). As such, it is open to determine whether there are (strongly) subquadratic algorithms for LCIS or whether such algorithms can be ruled out under SETH. The starting point of our work is to settle this question.

\subsection{Longest Common Increasing Subsequence (LCIS)}

The Longest Common Increasing Subsequence problem on $k$ sequences ($k$-LCIS) is defined as follows: Given integer sequences $X_1,\dots,X_k$ of length at most $n$, determine the length of the longest sequence $Z$ such that $Z$ is a strictly increasing sequence of integers and $Z$ is a subsequence of each $X_i, i\in\{1,\dots,k\}$. For $k=1$, we obtain the well-studied longest increasing subsequence problem (LIS; we refer to \cite{CrochemoreP10} for an overview), which has an $\Oh{n \log n}$ time solution and a matching lower bound in the decision tree model~\cite{Fredman75}. The extension to $k=2$, denoted simply as LCIS, has been proposed by Yang, Huang, and Chao~\cite{YangHC05}, partially motivated as a generalization of LIS and by potential applications in bioinformatics. They obtained an $\Oh{n^2}$ time algorithm, leaving open the natural question whether there exists a way to extend the near-linear time solution for LIS to a near-linear time solution for multiple sequences.

Interestingly, already a classic connection between LCS and LIS combined with a recent conditional lower bound of Abboud, Backurs and Vassilevska Williams~\cite{AbboudBVW15} yields a partial negative answer assuming SETH. 

\begin{observation}[Folklore reduction, implicit in \cite{HuntS77}, explicit in \cite{JacobsonV92}]\label{obs:lcs2lcis}
After $\Oh{kn^2}$ time preprocessing, we can solve $k$-LCS by a single call to $(k-1)$-LCIS on sequences of length at most $n^2$.
\end{observation}
\onlySHORT{A simple proof of this Observation can be found in the full version.}
\onlyFULL{
\begin{proof}
Let $\occ(\sigma)$ denote the descending sequence of positions $k$ with $X_1[k] = \sigma$. We define sequences $X_i' = \occ(X_i[0]) \cdots \occ(X_i[|X_i|-1])$ for all $i \in \{2,\dots,k\}$. It is straightforward to see that for any $\ell$, the length-$\ell$ increasing common subsequences of $X_2', \dots, X_k'$ are in one-to-one correspondence to length-$\ell$ common subsequences of $X_1,\dots,X_k$. Thus, the length of the LCIS of $X_2', \dots, X_k'$ is equal to the length of the LCS of $X_1,\dots,X_k$, and the claim follows since $|\occ(\sigma)| \le n$ for all $\sigma$.
\end{proof}}

\begin{corollary}
Unless SETH fails, there is no $\Oh{n^{\frac{3}{2}-\varepsilon}}$ time algorithm for LCIS for any constant $\varepsilon >0$.
\end{corollary}
\begin{proof}
Note that by the above reduction, an $\Oh{n^{\frac{3}{2}-\varepsilon}}$ time LCIS algorithm would give an $\Oh{n^{3-2\varepsilon}}$ time algorithm for 3-LCS. Such an algorithm would refute SETH by a result of Abboud et al.~\cite{AbboudBVW15}.
\end{proof}

While this rules out near-linear time algorithms, still an unsatisfying large polynomial gap between best upper and conditional lower bounds persists.

\subsection{Our Results}

Our first result is a tight SETH-based lower bound for LCIS.
\begin{theorem}\label{thm:lcis}
Unless SETH fails, there is no $\Oh{n^{2-\varepsilon}}$ time algorithm for LCIS for any constant $\varepsilon > 0$.
\end{theorem}
We extend our main result in several directions.

\subsubsection{Parameterized Complexity I: Solution Size} 

Subsequent work~\cite{ChanZFYZ07,KutzBKK11} improved over Yang et al.'s algorithm when certain input parameters are small. Here, we focus particularly on the solution size, i.e., the length $L$ of the LCIS. Kutz et al.~\cite{KutzBKK11} provided an algorithm running in time $\Oh{nL\log \log n + n\log n}$. If $L$ is small compared to its worst-case upper bound of $n$, say $L = n^{\frac{1}{2}\pm o(1)}$, this algorithm runs in strongly subquadratic time. Interestingly, exactly for this case, the reduction from 3-LCS to LCIS of Observation~\ref{obs:lcs2lcis} already yields a matching SETH-based lower bound of $(Ln)^{1-o(1)} = n^{\frac{3}{2}-o(1)}$. However, for smaller $L$, this reduction yields no lower bound at all and only a non-matching lower bound for larger $L$. We remedy this situation by the following result.\footnote{We mention in passing that a systematic study of the complexity of LCS in terms of such input parameters has been performed recently in~\cite{BringmannK17}.}

\begin{theorem}\label{thm:lcis-nL}
Unless SETH fails, there is no $\Oh{(nL)^{1-\varepsilon}}$ time algorithm for LCIS for any constant $\varepsilon>0$. This even holds restricted to instances with $L = n^{\gamma \pm o(1)}$, for \emph{arbitrarily chosen $0 < \gamma \le 1$}.
\end{theorem}

\subsubsection{Parameterized Complexity II: $k$-LCIS}

For constant $k\ge 3$, we can solve $k$-LCIS in $\Oh{n^k \polylog(n)}$ time~\cite{ChanZFYZ07,KutzBKK11}, or even $\Oh{n^k}$ time
\onlySHORT{(see the appendix in the full version).}
\onlyFULL{(see the appendix).}
While it is known that $k$-LCS cannot be computed in time $\Oh{n^{k-\varepsilon}}$ for any constant $\varepsilon >0, k\ge 2$ unless SETH fails~\cite{AbboudBVW15}, this does not directly transfer to $k$-LCIS, since the reduction in Observation~\ref{obs:lcs2lcis} is not tight. However, by extending our main construction, we can prove the analogous result.

\begin{theorem}\label{thm:klcis}
Unless SETH fails, there is no $\Oh{n^{k-\varepsilon}}$ time algorithm for $k$-LCIS for any constant $k\ge 3$ and $\varepsilon >0$. 
\end{theorem}

\subsubsection{Longest Common Weakly Increasing Subsequence (LCWIS)}

We consider a closely related variant of LCIS called the Longest Common Weakly Increasing Subsequence ($k$-LCWIS): Here, given integer sequences $X_1,\dots,X_k$ of length at most $n$, the task is to determine the longest \emph{weakly increasing} (i.e.~non-decreasing) integer sequence $Z$ that is a common subsequence of $X_1,\dots,X_k$. Again, we write LCWIS as a shorthand for 2-LCWIS.
Note that the seemingly small change in the notion of increasing sequence has a major impact on algorithmic and hardness results: Any instance of LCIS in which the input sequences are defined over a small-sized alphabet $\Sigma \subseteq \mathbb{Z}$, say $|\Sigma| = \Oh{n^{1/2}}$, can be solved in strongly subquadratic time $\Oh{nL \log n} = \Oh{n^{3/2} \log n}$~\cite{KutzBKK11}, by using the fact that $L \le |\Sigma|$. In contrast, LCWIS is quadratic-time SETH hard already over slightly superlogarithmic-sized alphabets~\cite{Polak17}. We give a substantially different proof for this fact and generalize it to $k$-LCWIS.

\begin{theorem}\label{thm:lcwis}
Unless SETH fails, there is no $\Oh{n^{k-\varepsilon}}$ time algorithm for $k$-LCWIS for any constant $k\ge 3$ and $\varepsilon >0$. This even holds restricted to instances defined over an alphabet of size $|\Sigma| \le f(n) \log n$ for any function $f(n) = \omega(1)$ growing arbitrarily slowly.
\end{theorem}

\subsubsection{Strengthening the Hardness}

In an attempt to strengthen the conditional lower bounds for Edit Distance and LCS~\cite{BackursI15,AbboudBVW15,BringmannK15}, particularly, to obtain barriers even for subpolynomial improvements, Abboud, Hansen, Vassilevska Williams, and Williams~\cite{AbboudHVWW16} gave the first fine-grained reductions from the satisfiability problem on branching programs. Using this approach, the quadratic-time hardness of a problem can be explained by considerably weaker variants of SETH, making the conditional lower bound stronger. We show that our lower bounds also hold under these weaker variants. In particular, we prove the following.

\begin{theorem}\label{thm:lcis-bp}
There is no strongly subquadratic time algorithm for LCIS, unless there is, for some $\eps > 0$, an $\Oh{(2-\eps)^N}$ algorithm for the satisfiability problem on branching programs of width $W$ and length $T$ on $N$ variables with $(\log W)(\log T) = \oh{N}$.
\end{theorem}

\subsection{Discussion, Outline and Technical Contributions}

%\subsubsection{Optimality of Dynamic-Programming Formulations}

Apart from an interest in LCIS and its close connection to LCS, our work is also motivated by an interest in the \emph{optimality of dynamic programming (DP) algorithms}\footnote{We refer to~\cite{ZhuWWW16} for a simple quadratic-time DP formulation for LCIS.}. Notably, many conditional lower bounds in $\mathsf{P}$ target problems with natural DP algorithms that are proven to be near-optimal under some plausible assumption (see, e.g., \cite{Bringmann14, AbboudVWW14, BackursI15, BackursI16, AbboudBVW15, BringmannK15, BackursT17, CyganMWW17, KunnemannPS17} and \cite{VassilevskaW15} for an introduction to the field). Even if we restrict our attention to problems that find optimal sequence alignments under some restrictions, such as LCS, Edit Distance and LCIS, the currently known hardness proofs differ significantly, despite seemingly small differences between the problem definitions. Ideally, we would like to classify the properties of a DP formulation which allow for matching conditional lower bounds.

One step in this direction is given by the \emph{alignment gadget framework}~\cite{BringmannK15}. Exploiting normalization tricks, this framework gives an abstract property of sequence similarity measures to allow for SETH-based quadratic lower bounds. Unfortunately, as it turns out, we cannot directly transfer the alignment gadget hardness proof for LCS to LCIS -- some indication for this difficulty is already given by the fact that LCIS can be solved in strongly subquadratic time over sublinear-sized alphabets~\cite{KutzBKK11}, while the LCS hardness proof already applies to binary alphabets. By collecting gadgetry needed to overcome such difficulties (that we elaborate on below), we hope to provide further tools to generalize more and more quadratic-time lower bounds based on SETH.

\subsubsection{Technical Challenges}

The known conditional lower bounds for global alignment problems such as LCS and Edit Distance work as follows. The reductions start from the quadratic-time SETH-hard Orthogonal Vectors problem (OV), that asks to determine, given two sets of $(0,1)$-vectors $\mathcal{U} = \{u_0, \ldots, u_{n-1}\}, \mathcal{V} = \{v_0, \ldots, v_{n-1}\} \subseteq \{0,1\}^d$ over $d=n^{o(1)}$ dimensions, whether there is a pair $i,j$ such that $u_i$ and $v_j$ are orthogonal, i.e., whose inner product $(u_i\dotp v_j) := \sum_{k=0}^{d-1} u_i[k]\cdot v_j[k]$ is 0 (over the integers). Each vector $u_i$ and $v_j$ is represented by a (normalized) vector gadget $\VG_\x(u_i)$ and $\VG_\y(v_j)$, respectively. Roughly speaking, these gadgets are combined to sequences $X$ and $Y$ such that each candidate for an optimal alignment of $X$ and $Y$ involves locally optimal alignments between $n$ pairs $\VG_\x(u_i), \VG_\y(v_j)$ -- the optimal alignment exceeds a certain threshold if and only if there is an orthogonal pair $u_i,v_j$.

An analogous approach does not work for LCIS: Let $\VG_\x(u_i)$ be defined over an alphabet $\Sigma$ and $\VG_\x(u_{i'})$ over an alphabet $\Sigma'$. If $\Sigma$ and $\Sigma'$ overlap, then $\VG_\x(u_i)$ and $\VG_\x(u_{i'})$ cannot both be aligned in an optimal alignment without interference with each other. On the other hand, if $\Sigma$ and $\Sigma'$ are disjoint, then each vector $v_j$ should have its corresponding vector gadget $VG_\y(v_j)$ defined over both $\Sigma$ and $\Sigma'$ to enable to align $\VG_\x(u_i)$ with $\VG_\y(v_j)$ as well as $\VG_\x(u_{i'})$ with $\VG_\y(v_j)$. The latter option drastically increases the size of vector gadgets. Thus, we must define all vector gadgets over a common alphabet $\Sigma$ and make sure that \emph{only a single pair} $\VG_\x(u_i),\VG_\y(v_j)$ is aligned in an optimal alignment (in contrast with $n$ pairs aligned in the previous reductions for LCS and Edit Distance).

\subsubsection{Technical Contributions and Proof Outline}

Fortunately, a surprisingly simple approach works:  As a key tool, we provide \emph{separator sequences} $\alpha_0\dots\alpha_{n-1}$ and $\beta_0\dots\beta_{n-1}$ with the following properties: (1) for every $i,j \in \{0,\dots,n-1\}$ the LCIS of $\alpha_0 \dots \alpha_i$ and $\beta_0 \dots \beta_j$ has a length of $f(i+j)$, where $f$ is a linear function, and (2) $\sum_i |\alpha_i|$ and $\sum_j |\beta_j|$ are bounded by $n^{1 + o(1)}$. Note that existence of such a gadget is somewhat unintuitive: condition (1) for $i=0$ and $j=n-1$ requires $|\alpha_0| = \Omega(n)$, yet still the total length $\sum_i |\alpha_i|$ must not exceed the length of $|\alpha_0|$ significantly. Indeed, we achieve this by a careful inductive construction that generates such sequences with heavily varying block sizes $|\alpha_i|$ and $|\beta_j|$.

We apply these separator sequences as follows. We first define simple vector gadgets $\VG_\x(u_i),\VG_\y(v_j)$ over an alphabet $\Sigma$ such that the length of an LCIS of $\VG_\x(u_i)$ and $\VG_\y(v_j)$ is $d-(u_i \dotp v_j)$. Then we construct the separator sequences as above over an alphabet $\Sigma_<$ whose elements are strictly smaller than all elements in $\Sigma$.  Furthermore, we create analogous separator sequences $\alpha'_0\dots\alpha'_{n-1}$ and $\beta_0'\dots\beta'_{n-1}$ which satisfy a property like (1) for all suffixes instead of prefixes, using an alphabet $\Sigma_>$ whose elements are strictly larger than all elements in $\Sigma$. Now, we define
\begin{align*}
X & = \alpha_0 \VG_\x(u_0) \alpha'_0  \dots \alpha_{n-1} \VG_\x(u_{n-1}) \alpha'_{n-1}, \\
Y & = \beta_0 \VG_\y(v_0) \beta'_0  \dots \beta_{n-1} \VG_\y(v_{n-1}) \beta'_{n-1}.
\end{align*}
As we will show in Section~\ref{sec:lcis}, the length of an LCIS of $X$ and $Y$ is $C - \min_{i,j} (u_i \dotp v_j)$ for some constant $C$ depending only on $n$ and $d$.

In contrast to previous such OV-based lower bounds, we use heavily varying separators (paddings) between vector gadgets.

\onlySHORT{
\subsection{Paper organization}

After setting up conventions and introducing our hardness assumptions in Section~\ref{sec:prelim}, we give the main construction, i.e., the hardness of LCIS in Section~\ref{sec:lcis}. The proofs of Theorems~\ref{thm:klcis}, \ref{thm:lcis-nL}, \ref{thm:lcwis} and \ref{thm:lcis-bpseth} can be found in the full version in the appendix. We conclude with some open problems in Section~\ref{sec:conclusion}.
}

\section{Preliminaries}
\label{sec:prelim}

As a convention, we use capital or Greek letters to denote sequences over integers. Let $X,Y$ be integer sequences. We write $|X|$ for the length of $X$, $X[k]$ for the $k$-th element in the sequence $X$ ($k\in\{0,\ldots,|X|-1\}$), and $X\conc Y = XY$ for the concatenation of $X$ and $Y$. We say that $Y$ is a subsequence of $X$ if there exist indices $0\le i_1 < i_2 < \cdots < i_{|Y|}\le |X| - 1$  such that $X[i_k] = Y[k]$ for all $k\in \{0,\dots,|Y|-1\}$. Given any number of sequences $X_1,\dots,X_k$, we say that $Y$ is a common subsequence of $X_1,\dots,X_k$ if $Y$ is a subsequence of each $X_i, i\in \{1,\dots,k\}$. $X$ is called strictly increasing (or weakly increasing) if $X[0] < X[1] < \cdots < X[|X|-1]$ (or $X[0] \le X[1] \le \cdots \le X[|X|-1]$). For any $k$ sequences $X_1, \ldots, X_k$, we denote by $\lcis(X_1, \ldots, X_k)$ the length of their longest common subsequence that is strictly increasing.

\onlySHORT{\textbf{Hardness Assumptions.}}
\onlyFULL{\subsection{Hardness Assumptions}}
All of our lower bounds hold assuming the Strong Exponential Time Hypothesis (SETH), introduced by Impagliazzo and Paturi~\cite{ImpagliazzoP01,ImpagliazzoPZ01}. It essentially states that no exponential speed-up over exhaustive search is possible for the CNF satisfiability problem.

\begin{hypothesis}[Strong Exponential Time Hypothesis (SETH)]
There is no $\varepsilon > 0$ such that for all $d \ge 3$ there is an $\Oh{2^{(1-\varepsilon)n}}$ time algorithm for $d$-SAT.
\end{hypothesis}

This hypothesis implies tight hardness of the $k$-Orthogonal Vectors problem ($k$-OV), which will be the starting point of our reductions: Given $k$ sets $\mathcal{U}_1, \dots, \mathcal{U}_k \subseteq \{0,1\}^d$, each with $|\mathcal{U}_i| = n$ vectors over $d= n^{o(1)}$ dimensions, determine whether there is a $k$-tuple $(u_1, \dots, u_k) \in \mathcal{U}_1 \times \cdots \times \mathcal{U}_k$ such that $\sum_{\ell=0}^{d-1} \prod_{i=1}^k u_i[\ell] = 0$. By exhaustive enumeration, it can be solved in time $\Oh{n^k d} = n^{k+o(1)}$. The following conjecture is implied by SETH by the well-known split-and-list technique of Williams~\cite{Williams05} (and the sparsification lemma~\cite{ImpagliazzoPZ01}).

\begin{hypothesis}[$k$-OV conjecture]
\label{hyp:kov}
Let $k\ge 2$. There is no $\Oh{n^{k-\varepsilon}}$ time algorithm for $k$-OV, with $d= \omega(\log n)$, for any constant $\varepsilon > 0$.
\end{hypothesis}

For the special case of $k=2$, which we simply denote by OV, we obtain the following weaker conjecture.
\begin{hypothesis}[OV conjecture]
\label{hyp:ov}
There is no $\Oh{n^{2-\varepsilon}}$ time algorithm for OV, with $d=\omega(\log n)$, for any constant $\varepsilon > 0$. Equivalently, even restricted to instances with $|\mathcal{U}_1| = n$ and $|\mathcal{U}_2| = n^{\gamma}$, $0 < \gamma \le 1$, there is no $\Oh{n^{1+\gamma- \varepsilon}}$ time algorithm for OV, with $d= \omega(\log n)$, for any constant $\varepsilon > 0$.
\end{hypothesis}

A proof of the folklore equivalence of the statements for equal and unequal set sizes can be found, e.g., in~\cite{BringmannK15}.

\section{Main Construction: Hardness of LCIS}
\label{sec:lcis}

In this section, we prove quadratic-time SETH hardness of LCIS, i.e., prove Theorem~\ref{thm:lcis}. We first introduce an \emph{inflation} operation, which we then use to construct our separator sequences. After defining simple vector gadgets, we show how to embed an Orthogonal Vectors instance using our vector gadgets and separator sequences. 

\subsection{Inflation}

We begin by introducing the inflation operation, which roughly corresponds to weighing the sequences.

\begin{definition}
For a sequence $A = \seq{a_0, a_1, \ldots, a_{n-1}}$ of integers we define:
\[\infl(A) = \seq{2a_0-1, 2a_0, 2a_1-1, 2a_1, \ldots, 2a_{n-1}-1, 2a_{n-1}}.\]
\end{definition}

\begin{lemma}
\label{le:infl}
 For any two sequences $A$ and $B$, $\lcis(\infl(A), \infl(B)) = 2 \cdot \lcis(A, B)$.

 \begin{proof}
  Let $C$ be the longest common increasing subsequence of
  $A$ and $B$. Observe that $\infl(C)$ is a common increasing subsequence of
  $\infl(A)$ and $\infl(B)$ of length $2 \cdot |C|$, thus
  $\lcis(\infl(A), \infl(B)) \geq 2 \cdot \lcis(A, B)$. 

  Conversely, let $\bar A$ denote $\infl(A)$ and $\bar B$ denote $\infl(B)$.
  Let $\bar C$ be the longest common increasing subsequence of $\bar A$ and $\bar B$.
  If we divide all elements of $\bar C$ by $2$ and round up to the closest integer,
  we end up with a weakly increasing sequence. Now, if we remove duplicate elements
  to make this sequence strictly increasing, we obtain $C$, a common increasing
  subsequence of $A$ and $B$.
  At most $2$ distinct elements may become equal after division by $2$ and rounding,
  therefore $C$ contains at least $\ceil{\lcis(\bar A, \bar B) / 2}$ elements,
  so $2 \cdot \lcis(A, B) \geq \lcis(\bar A, \bar B)$.
  This completes the proof.
 \end{proof}

\end{lemma}

\subsection{Separator sequences}
\label{sec:sepseq}

\begin{figure}[!b]
\begin{center}
\includegraphics[width=0.67\textwidth]{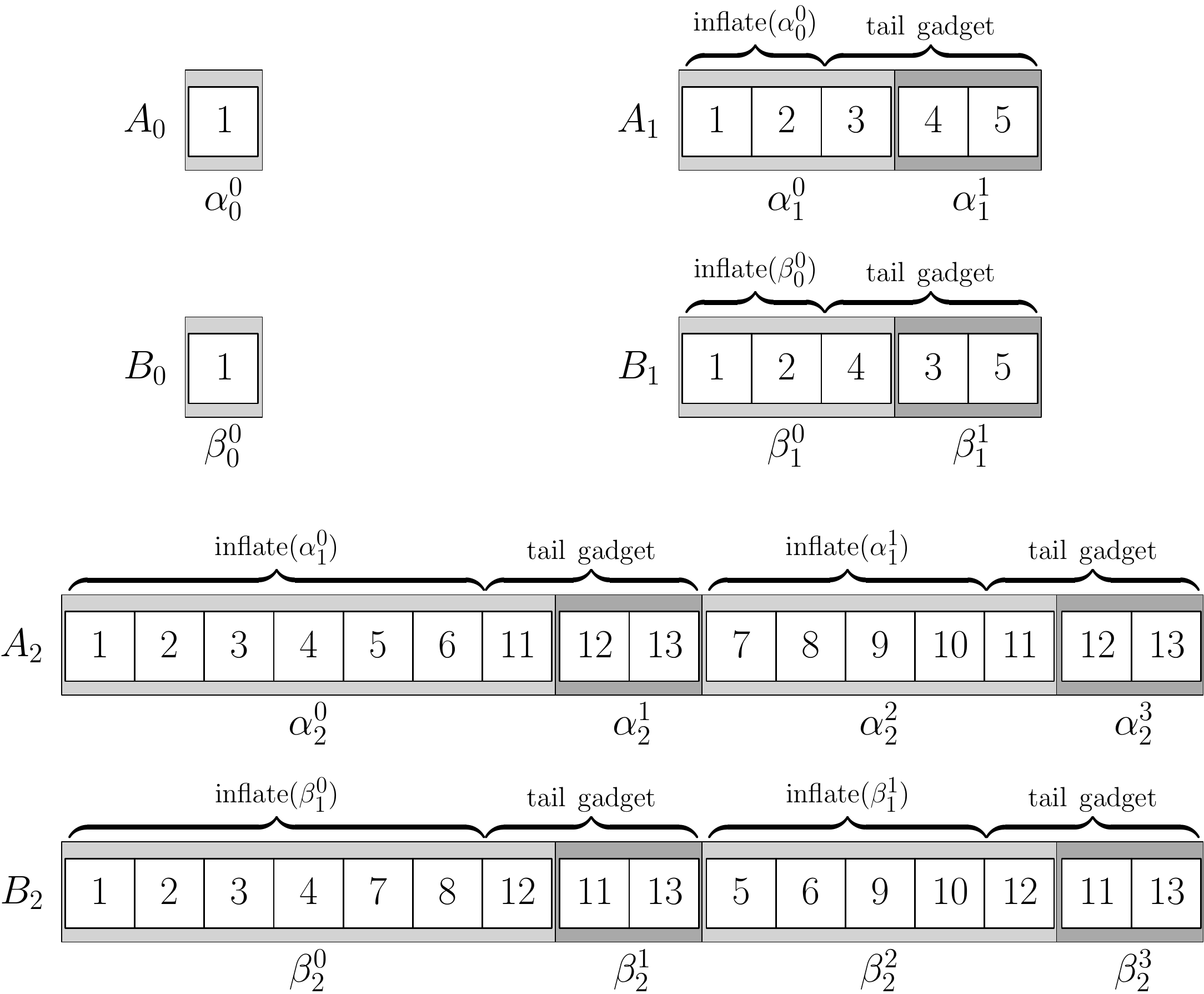}
\end{center}
\caption{
  Initial steps of the inductive construction of the separator sequences.
}
\end{figure}

Our goal is to construct two sequences $A$ and $B$ which can be split into $n$
blocks, i.e. $A=\alpha_0\alpha_1\ldots\alpha_{n-1}$ and
$B=\beta_0\beta_1\ldots\beta_{n-1}$, such that the length of the longest common
increasing subsequence of the first $i$ blocks of $A$ and the first $j$ blocks
of $B$ equals $i + j$, up to an additive constant.
We call $A$ and $B$ \emph{separator sequences}, and use them later to separate
vector gadgets in order to make sure that only one pair of gadgets may interact
with each other at the same time.

We construct the separator sequences inductively. For every $k \in \Nat$, the
sequences $A_k$ and $B_k$ are concatenations of $2^k$ blocks (of varying sizes),
$A_k = \alpha_k^0\alpha_k^1\ldots\alpha_k^{2^k-1}$ and
$B_k = \beta_k^0\beta_k^1 \ldots\beta_k^{2^k-1}$.
Let $s_k$ denote the largest element of both sequences. As we will soon observe,
$s_k = 2^{k+2} - 3$.

The construction works as follows: for $k = 0$, we can simply set $A_0$ and $B_0$
as one-element sequences~$\seq{1}$. We then construct $A_{k+1}$ and $B_{k+1}$
inductively from $A_k$ and $B_k$ in two steps. First, we inflate both $A_k$ and
$B_k$, then after each (now inflated) block we insert $3$-element sequences,
called \emph{tail gadgets}, $\seq{2s_k+2, 2s_k+1, 2s_k+3}$ for $A_{k+1}$ and
$\seq{2s_k+1, 2s_k+2, 2s_k+3}$ for $B_{k+1}$.
Formally, we describe the construction by defining blocks of the new sequences.
For $i\in\{0,1,\ldots,2^k-1\}$,
\begin{align*}
\alpha_{k+1}^{2i} &= \infl(\alpha_k^i) \conc \seq{2s_k + 2}, &
\alpha_{k+1}^{2i+1} &= \seq{2s_k+1, 2s_k+3}, \\
\beta_{k+1}^{2i} &= \infl(\beta_k^i) \conc \seq{2s_k + 1}, &
\beta_{k+1}^{2i+1} &= \seq{2s_k+2, 2s_k+3}.
\end{align*}
Note that the symbols appearing in tail gadgets do not appear in the inflated sequences.
The largest element of both new sequences $s_{k+1}$ equals $2 s_k + 3$,
and solving the recurrence gives indeed $s_k = 2^{k+2} - 3$.

\begin{figure}
\begin{center}
\includegraphics[width=0.5\textwidth]{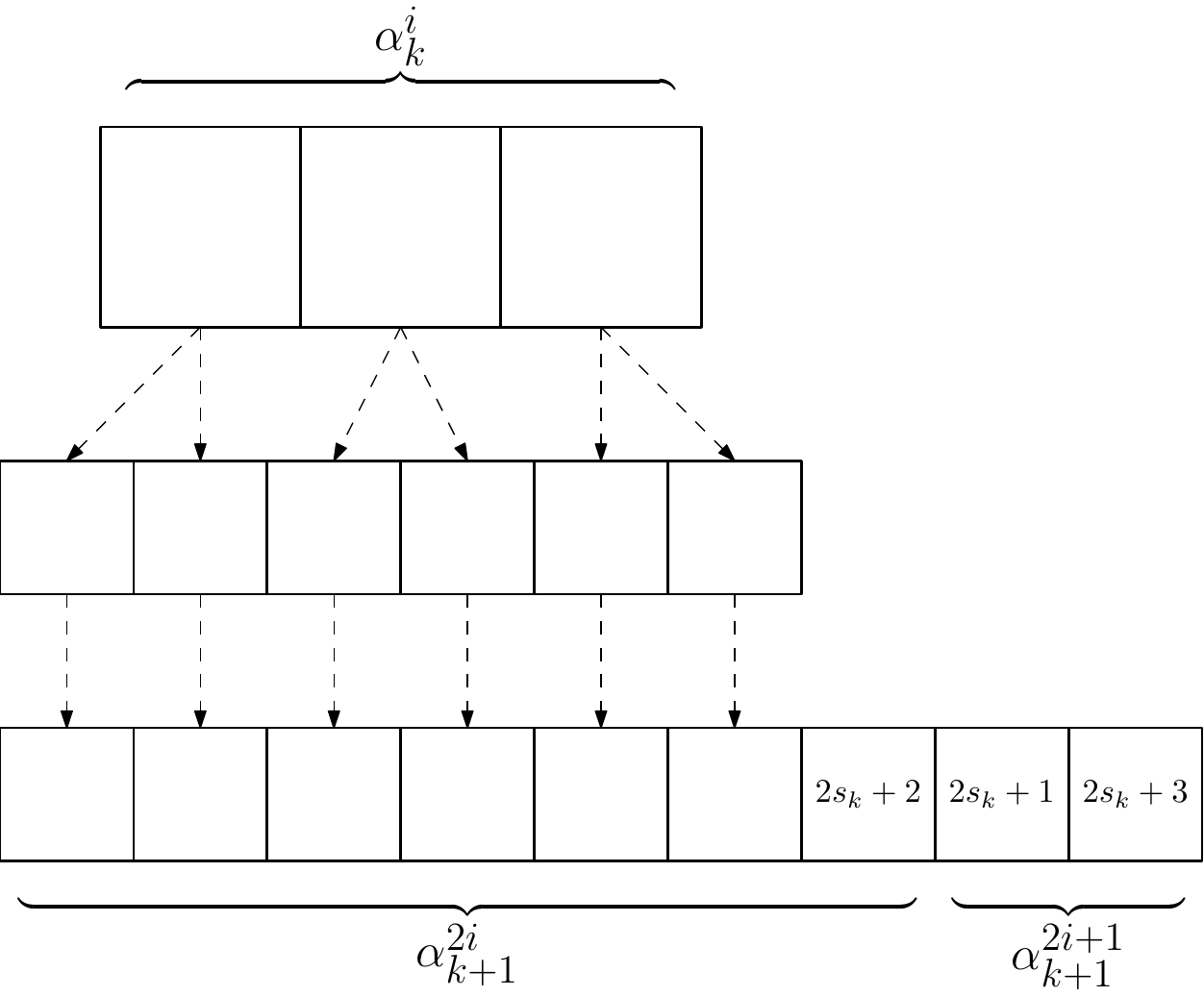}
\hfill
\includegraphics[width=0.45\textwidth]{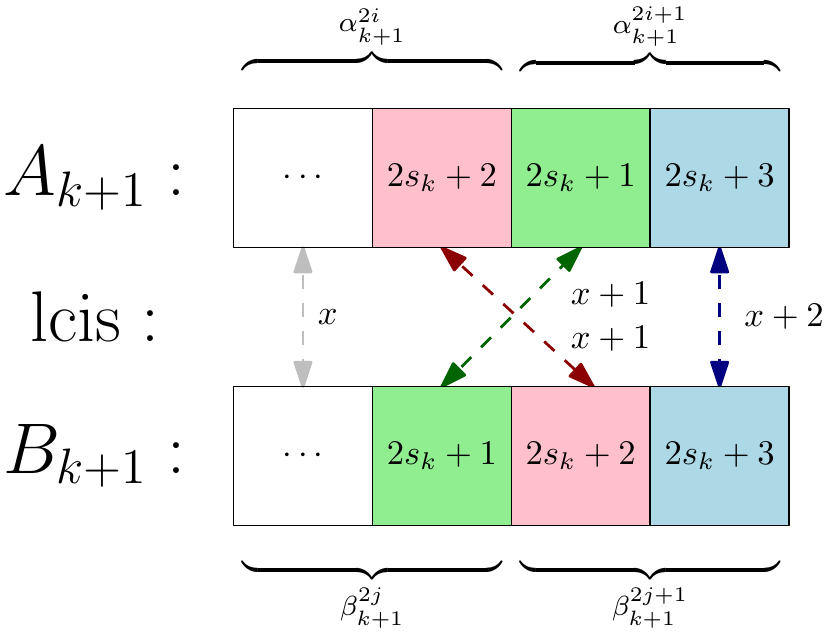}
\end{center}
\caption{
  Left: constructing $A_{k+1}$ from $A_k$.
  Right: intuition behind tail gadgets, $x=2i+2j+2^{k+1}$.
}
\end{figure}

Now, let us prove two useful properties of the separator sequences.

\begin{lemma}
  \label{le:sep-len}
  $|A_k| = |B_k| = \brac{\frac{3}{2}k+1} \cdot 2^k$ = $\Oh{k2^k}$.
  \begin{proof}
  Observe that $|A_{k+1}| = 2|A_k| + 3 \cdot 2^k$. Indeed, to obtain $A_{k+1}$
  first we double the size of $A_k$ and then add $3$ new elements for each of
  the $2^k$ blocks of $A_k$. Solving the recurrence completes the proof.
  The same reasoning applies to $B_k$.
  \end{proof}
\end{lemma}

\begin{lemma}
  \label{le:sep-lcis}
  For every $i, j \in \set{0, 1, \ldots, 2^k-1}$,  $\lcis(\alpha_k^0\ldots\alpha_k^i, \beta_k^0\ldots\beta_k^j) = i + j + 2^k$.

  \begin{proof}
  The proof is by induction on $k$.
  Assume the statement is true for $k$ and let us prove it for $k+1$.

  \paragraph{The ``$\geq$'' direction.}
  First, consider the case when both $i$ and $j$ are even. Observe that
  $\infl(\alpha_k^0\ldots\alpha_k^{i/2})$ and
  $\infl(\beta_k^0\ldots\beta_k^{j/2})$ are subsequences of
  $\alpha_{k+1}^0\ldots\alpha_{k+1}^i$ and $\beta_{k+1}^0\ldots\beta_{k+1}^j$,
  respectively. Thus, using the induction hypothesis and inflation properties,
  \begin{multline*}
  \lcis(\alpha_{k+1}^0\ldots\alpha_{k+1}^i$, $\beta_{k+1}^0\ldots\beta_{k+1}^j)
   \geq \lcis(\infl(\alpha_k^0\ldots\alpha_k^{i/2}),
        \infl(\beta_k^0\ldots\beta_k^{j/2})) = \\
   = 2 \cdot \lcis(\alpha_k^0\ldots\alpha_k^{i/2}, \beta_k^0\ldots\beta_k^{j/2})
   = 2 \cdot (i/2 + j/2 + 2^k) = i + j + 2^{k+1}.
  \end{multline*}
  If $i$ is odd and $j$ is even, refer to the previous case to get
  a common increasing subsequence of $\alpha_{k+1}^0\ldots\alpha_{k+1}^{i-1}$
  and $\beta_{k+1}^0\ldots\beta_{k+1}^j$ of length $i - 1 + j + 2^{k+1}$
  consisting only of elements less than or equal to $2s_k$, and append the element
  $2s_k+1$ to the end of it.
  Analogously, for $i$ even and $j$ odd, take such an LCIS of
  $\alpha_{k+1}^0\ldots\alpha_{k+1}^i$ and
  $\beta_{k+1}^0\ldots\beta_{k+1}^{j-1}$, and append $2s_k+2$.
  Finally, for both $i$ and $j$ odd, take an LCIS of
  $\alpha_{k+1}^0\ldots\alpha_{k+1}^{i-1}$ and
  $\beta_{k+1}^0\ldots\beta_{k+1}^{j-1}$, and append $2s_k+1$ and $2s_k+3$.

  \paragraph{The ``$\leq$'' direction.}
  We proceed by induction on $i + j$.
  Fix $i$ and $j$, and let $L$ be a longest common increasing subsequence of
  $\alpha_{k+1}^0\ldots\alpha_{k+1}^i$ and $\beta_{k+1}^0\ldots\beta_{k+1}^j$.

  If the last element of $L$ is less than or equal to $2s_k$, $L$ is in fact
  a common increasing subsequence of
  $\infl(\alpha_k^0\ldots\alpha_k^{\floor{i/2}})$ and
  $\infl(\beta_k^0\ldots\beta_k^{\floor{j/2}})$, thus, by the induction
  hypothesis and inflation properties, $|L| \leq 2 \cdot (\floor{i/2} + \floor{j/2} + 2^k) \leq i + j + 2^{k+1}$.

  The remaining case is when the last element of $L$ is greater than $2s_k$. In this case, consider the second-to-last element of $L$. It must belong to some blocks
  $\alpha_{k+1}^{i'}$ and $\beta_{k+1}^{j'}$ for $i' \leq i$ and $j' \leq j$, and we claim that $i=i'$ and $j=j'$ cannot hold simultaneously:
  by construction of separator sequences, if blocks $\alpha_{k+1}^i$ and $\beta_{k+1}^j$ have a common
  element larger than $2s_k$, then it is the only common element of these two
  blocks. Therefore, it cannot be the case that both $i=i'$ and $j=j'$, because
  the last two elements of $L$ would then be located in $\alpha_{k+1}^i$ and
  $\beta_{k+1}^j$. As a consequence, $i'+j' < i+j$, which lets us apply the
  induction hypothesis to reason that the prefix of $L$ omitting its last
  element is of length at most $i' + j' + 2^{k+1}$. Therefore, $|L| \leq 1 + i' + j' + 2^{k+1} \leq i + j + 2^{k+1}$, which completes the proof.
  \end{proof}
\end{lemma}

Observe that if we reverse the sequences $A_k$ and $B_k$ along with changing all elements to their negations, i.e.~$x$ to $-x$, we obtain sequences $\hat A_k$ and $\hat B_k$ such that
$\hat A_k$ splits into $2^k$ blocks
$\hat \alpha_k^0 \ldots \hat \alpha_k^{2^k-1}$,
$\hat B_k$ splits into $2^k$ blocks
$\hat \beta_k^0 \ldots \hat \beta_k^{2^k-1}$,
and
\begin{equation}
\label{eq:hat}
\lcis(\hat\alpha_k^i\ldots\hat\alpha_k^{2^k-1},
      \hat\beta_k^j\ldots\hat\beta_k^{2^k-1})
      = 2 \cdot (2^k - 1) - i - j + 2^k.
\end{equation}

Finally, observe that we can add any constant to all elements of the sequences
$A_k$ and $B_k$ (as well as $\hat A_k$ and $\hat B_k$) without changing the
property stated in Lemma~\ref{le:sep-lcis} (and its analogue for $\hat A_k$ and
$\hat B_k$, i.e.~Equation~\eqref{eq:hat}).

\subsection{Vector gadgets}
\label{sec:vg}
Let $\mathcal{U} = \{u_0, \ldots, u_{n-1}\}$ and $\mathcal{V} = \{v_0, \ldots, v_{n-1}\}$
be two sets of $d$-dimensional $(0,1)$-vectors.

For $i\in\{0,1,\ldots,n-1\}$ let us construct the vector gadgets $U_i$ and $V_i$ as $2d$-element sequences, by defining, for every $p \in \{0, 1, \ldots, d-1\}$,
\begin{align*}
(U_i[2p-1], U_i[2p]) &= \begin{cases}
                           (2p - 1, 2p)\ & \mbox{if}\ u_i[p] = 0, \\
                           (2p - 1, 2p -1)\ & \mbox{if}\ u_i[p] = 1,
                        \end{cases}\\
(V_i[2p-1], V_i[2p]) &= \begin{cases}
                 \mathrlap{(2p, 2p-1)}\hphantom{
                           (2p - 1, 2p -1)}\ & \mbox{if}\ v_i[p] = 0, \\
                           (2p, 2p)\ & \mbox{if}\ v_i[p] = 1.
                        \end{cases}
\end{align*}

Observe that at most one of the elements $2p-1$ and $2p$ may appear in the LCIS
of $U_i$ and $V_j$, and it happens if and only if $u_i[p]$ and $v_j[p]$ are not
both equal to one. Therefore, $\lcis(U_i, V_j) = d - (u_i \dotp v_j)$, and, in
particular, $\lcis(U_i, V_j) = d$ if and only if $u_i$ and $v_j$ are orthogonal.

\subsection{Final construction}
\label{sec:2-lcis-final}

To put all the pieces together, we plug vector gadgets $U_i$ and $V_j$ into the separator sequences from Section \ref{sec:sepseq}, obtaining two sequences whose LCIS depends on the minimal inner product of vectors $u_i$ and $v_j$. We provide a general construction of such sequences,
\onlySHORT{which will be useful for proving further results in the full version of the paper.}%
\onlyFULL{which will be useful in later sections.}

\begin{lemma}
\label{lem:2-lcis-main}
Let $X_0, X_1, \ldots, X_{n-1}$, $Y_0, Y_1, \ldots, Y_{n-1}$ be integer sequences such that none of them has an increasing subsequence longer than $\delta$. Then there exist sequences $X$ and $Y$ of length $\Oh{\delta \cdot n \log n} + \sum |X_i| + \sum |Y_j|$, constructible in linear time, such that:
\[\lcis(X,Y) = \max_{i,j} \lcis(X_i,Y_j) + C \]
for a constant $C$ that only depends on $n$ and $\delta$ and is $\Oh{n\delta}$.
 
\begin{proof}
We can assume that $n = 2^k$ for some positive integer $k$, adding some dummy sequences if necessary.
Recall the sequences $A_k$, $B_k$, $\hat A_k$ and $\hat B_k$ constructed in Section \ref{sec:sepseq}. Let $A, B, \hat A, \hat B$ be the sequences obtained from $A_k, B_k, \hat A_k, \hat B_k$ by applying inflation
$\ceil{\log_2\delta}$ times (thus increasing their length by a factor of
$\ell = 2^{\ceil{\log_2\delta}} \geq \delta$). Each of these four sequences splits
into (now inflated) blocks, e.g.~$A = \alpha_0 \alpha_1 \ldots \alpha_{n-1}$,
where $\alpha_i = \infl^{\ceil{\log_2\delta}}(\alpha_k^i)$.

We subtract from $A$ and $B$ a constant large enough for all their elements to be smaller than all elements of every $X_i$ and $Y_j$. Similarly, we add to $A'$ and $B'$ a constant large enough for all their elements to be larger than all elements of every $X_i$ and $Y_j$. Now, we can construct the sequences $X$ and $Y$ as follows:
\begin{align*}
X &= \alpha_0 X_0 \hat\alpha_0 \alpha_1 X_1 \hat\alpha_1 \ldots \alpha_{n-1} X_{n-1} \hat\alpha_{n-1},\\
Y &= \beta_0 Y_0 \hat\beta_0 \beta_1 Y_1 \hat\beta_1 \ldots \beta_{n-1} Y_{n-1} \hat\beta_{n-1}.
\end{align*}
We claim that
  \[\lcis(X,Y) = \ell \cdot (4 n - 2) + M\mbox{, where }M = \max_{i,j} {\lcis(X_i,Y_j)}. \]
 
  Let $X_i$ and $Y_j$ be the pair of sequences achieving $\lcis(X_i, Y_j) = M$. Recall that
  $\lcis(\alpha_0\ldots\alpha_i, \beta_0\ldots\beta_j) = \ell \cdot (i + j + n)$,
  with all the elements of this common subsequence preceding the elements of
  $X_i$ and $Y_j$ in $X$ and $Y$, respectively, and being smaller than them.
  In the same way $\lcis(\hat\alpha_i\ldots\hat\alpha_{n-1}, \hat\beta_j\ldots\hat\beta_{n-1})
  = \ell \cdot (2 \cdot (n - 1) - (i + j) + n)$ with all the elements
  of LCIS being greater and appearing later than those of $X_i$ and $Y_j$. By
  concatenating these three sequences we obtain a common increasing subsequence
  of $X$ and $Y$ of length $\ell \cdot (4 n - 2) + M$. 

\onlySHORT{
We defer the simple remainder of the proof, i.e., proving $\lcis(X,Y) \leq \ell \cdot (4 n - 2) + M$ to the full version of the paper.}
\onlyFULL{
  It remains to prove $\lcis(X,Y) \leq \ell \cdot (4 n - 2) + M$.
  Let $L$ be any common increasing subsequence of $X$ and $Y$. Observe that $L$
  must split into three (some of them possibly empty) parts $L = S G \hat S$
  with $S$ consisting only of elements of $A$ and $B$, $G$ -- only elements of
  $X_i$ and $Y_j$, and $\hat S$ -- elements of $\hat A$ and $\hat B$.

  Let $x$ be the last element of $S$ and $\hat x$ the first element of $\hat S$. We know
  that $x$ belongs to some blocks $\alpha_i$ of $A$ and $\beta_j$ of $B$, and
  $\hat x$ belongs to some blocks $\hat\alpha_{\hat i}$ of $\hat A$ and $\hat\beta_{\hat j}$ of $\hat B$.
  Obviously $i \leq \hat i$ and $j \leq \hat j$. By Lemma~\ref{le:sep-lcis} and
  inflation properties we have $|S| \leq \ell \cdot (i + j + n)$ and
  $|\hat S| \leq \ell \cdot (2 \cdot (n - 1) - (\hat i + \hat j) + n)$. We consider two cases:

  \textbf{Case 1.}
 If $i = \hat i$ and $j = \hat j$, then $G$ may only contain
    elements of $X_i$ and $Y_j$. Therefore
    \[ |L| \leq |S| + \lcis(X_i, Y_j) + |\hat S| \leq \ell \cdot (4n-2) + M.\]

  \textbf{Case 2.}
 If $i < \hat i$ or $j < \hat j$, then $G$ must be a strictly increasing subsequence of both $X_i \conc \cdots \conc X_{\hat i}$ and $Y_j \conc \cdots \conc Y_{\hat j}$ therefore its length can be bounded by \begin{multline*} |G|\leq\min(\delta\cdot(\hat i-i+1),\delta\cdot(\hat j-j+1)) \leq \ell\cdot(\min(\hat i-i,\hat j-j)+1) \leq \\ \leq \ell\cdot(\min(\hat i-i,\hat j-j)+\max(\hat i-i,\hat j-j)) = \ell \cdot(\hat i - i + \hat j - j).\end{multline*} On the other hand, $|S|+|\hat S|\leq \ell \cdot (4n-2-(\hat i-i)-(\hat j-j))$. From that we obtain $|L| \leq \ell \cdot (4n-2)$, as desired.
 
}
 \end{proof}

\end{lemma}

\onlyFULL{ We are ready to prove the main result of the paper.}
\begin{proof}[Proof of Theorem~\ref{thm:lcis}]
Let $\mathcal{U} = \{u_0, \ldots, u_{n-1}\}$, $\mathcal{V} = \{v_0, \ldots, v_{n-1}\}$ be two sets of $d$-dimensional binary vectors. In Section~\ref{sec:vg} we constructed vector gadgets $U_i$ and $V_j$, for $i, j \in \{0, 1,\ldots, n-1\}$, such that $\lcis(U_i,V_j) = d - (u_i \dotp v_j)$. To these sequences we apply Lemma~\ref{lem:2-lcis-main}, with $\delta = 2d$, obtaining sequences $X$ and $Y$ of length $\Oh{n \log n \poly(d)}$ such that $\lcis(X,Y) = C + d - \min_{i,j} (u_i \dotp v_j)$ for a constant $C$. This reduction, combined with an $\Oh{n^{2-\eps}}$ time algorithm for LCIS, would yield an $\Oh{n^{2-\eps} \polylog(n) \poly(d)}$ algorithm for OV, refuting Hypothesis~\ref{hyp:ov} and, in particular, SETH.
\end{proof}

%\begin{remark}
\onlyFULL{
With the reduction above, one can not only determine whether there exist a pair of orthogonal vectors or not, but also, in the latter case, calculate the minimum inner product over all pairs of vectors. Formally, by the above construction, we can reduce even the Most Orthogonal Vectors problem, as defined in Abboud et al.~\cite{AbboudBVW15} to LCIS. This bases hardness of LCIS already on the inability to improve over exhaustive search for the MAX-CNF-SAT problem, which is a slightly weaker conjecture than SETH.}
%\end{remark}

\onlyFULL{

\section{Matching Lower Bound for Output-Dependent Algorithms}

To prove our bivariate conditional lower bound of $(nL)^{1-o(1)}$, we provide a reduction from an OV instance with unequal vector set sizes.

\begin{proof}[Proof of Theorem~\ref{thm:lcis-nL}]
Let $0 < \gamma \le 1$ be arbitrary and consider any OV instance with sets $\mathcal{U}, \mathcal{V} \subseteq \{0, 1\}^d$ with $|\mathcal{U}| = n$, $|\mathcal{V}| = m = n^{\gamma}$ and $d=n^{o(1)}$. We reduce this problem, in linear time in the output size, to an LCIS instance with sequences $X$ and $Y$ satisfying $|X| = |Y| = \Oh{nd \log n}$ and an LCIS of length $\Oh{n^{\gamma}d}$. Theorem~\ref{thm:lcis-nL} is an immediate consequence of the reduction: an $\Oh{(nL)^{1-\varepsilon}}$ time LCIS algorithm would yield an OV algorithm running in time $\Oh{n^{1+\gamma-\eps'}}$, which would refute Hypothesis~\ref{hyp:ov} and, in particular, SETH.

It remains to show the reduction itself. Let $\mathcal{U} = \{u_0, \ldots, u_{n-1}\}$ and $\mathcal{V} = \{v_0, \ldots, v_{m-1}\}$
be two sets of $d$-dimensional $(0,1)$-vectors. By adding dummy vectors, we can assume without loss of generality that $n = q \cdot m$ for some integer $q$.

We use the vector gadgets $U_i$ and $V_j$ from Section \ref{sec:2-lcis-final}. This time, however, we group together every $q$ consecutive gadgets, i.e.,  $(U_0, \ldots, U_{q-1})$, $(U_{q}, \ldots, U_{2q-1})$,  and so on. Specifically, let $U_i^{[r]}$ be the $i$-th vector gadget shifted by an integer $r$ (i.e. with $r$ added to all its elements). We define, for each $l \in \{0, 1, \ldots, m-1\}$,
\[ \bar U_l = U_{lq}^{[2qd]} U_{lq+1}^{[2qd-2d]} \ldots U_{lq+q-1}^{[2d]}. \]
In a similar way, for $j \in \{0, 1, \ldots, m-1\}$, we replicate every $V_j$ gadget $q$ times with appropriate shifts, i.e., 
\[ \bar V_j = V_{j}^{[2qd]} V_{j}^{[2qd-2d]} \ldots V_{j}^{[2d]}. \]

Let us now determine $\lcis(\bar U_l, \bar V_j)$. No two gadgets grouped in $\bar U_l$ can contribute to an LCIS together, as the later one would have smaller elements. Therefore, only one $U_i$ gadget can be used, paired with the one copy of $V_j$ having the matching shift. This yields $\lcis(\bar U_l, \bar V_j) = \max_{lq \leq i < lq+q} \lcis (U_i, V_j)$, and in turn, also $\max_{l, j} \lcis (\bar U_l, \bar V_j) = \max_{i, j} \lcis (U_i, V_j) = d - \min_{i, j} (u_i \dotp v_j)$. 

Observe that every $\bar U_l$ is a concatenation of several $U_i$ gadgets, each one shifted to make its elements smaller than previous ones. Therefore, any increasing subsequence of $\bar U_l$ must be contained in a single $U_i$, and thus cannot be longer than $2d$. The same argument applies to every $\bar V_j$. Therefore, we can apply Lemma \ref{lem:2-lcis-main}, with $\delta = 2d$, to these sequences, obtaining $\bar X$ and $\bar Y$ satisfying:

\[ \lcis(\bar X,\bar Y) = C + d - \min_{i, j} (u_i \dotp v_j). \]

Recall that $C$ is some constant dependent only on $m$ and $d$, and $C = \Oh{md}$. The length of both $\bar X$ and $\bar Y$ is $\Oh{d m \log m + m q d} = \Oh{n d \log n}$, and the length of the output is $\Oh{md}$, as desired. 
\end{proof}

\section{Hardness of $k$-LCIS}

In this section we show that, assuming SETH, there is no $\Oh{n^{k-\eps}}$
algorithm for the $k$-LCIS problem, i.e., we prove Theorem~\ref{thm:klcis}.
To obtain this lower bound we show a reduction from the $k$-Orthogonal Vectors
problem (for definition, see Section~\ref{sec:prelim}).
There are two main ingredients of the reduction, i.e.~separator sequences and
vector gadgets, and both of them can be seen as natural generalizations of those
introduced in Section~\ref{sec:lcis}.

\subsection{Generalizing separator sequences}

Please note that in this section we use a notation
which is not consistent with the one from Section~\ref{sec:lcis},
because it has to accommodate indexing over $k$ sequences.

The aim of this section is to show, for any $N$ that is a power of two, how to construct $k$ sequences $A_1, A_2, \ldots, A_k$ such that each of them can be
split into $N$ blocks, i.e.~$A_i = \alpha_i^0\alpha_i^1\ldots\alpha_i^{N-1}$,
and for any choice of $j_1, j_2, \ldots, j_k \in \set{0, 1, \ldots, N-1}$
\begin{equation}
\label{eq:k-sep}
\lcis(\alpha_1^0\ldots\alpha_1^{j_1},
      \alpha_2^0\ldots\alpha_2^{j_2},
      \ldots,
      \alpha_k^0\ldots\alpha_k^{j_k}) = j_1 + j_2 + \cdots + j_k + N.
\end{equation}

As before, we construct separator sequences inductively, doubling the number of blocks in each step. Again, for $N=1$, we define the sequences by $A_i = \seq{1}, i \in \{1,\dots,k\}$. 

Suppose we have $N$-block sequences
$A_1, A_2, \ldots, A_k$, $A_i = \alpha_i^0\alpha_i^1\ldots\alpha_i^{N-1}$ as above. We show how to construct $2N$-block sequences
$B_1, B_2, \ldots, B_k$, $B_i = \beta_i^0\beta_i^1\ldots\beta_i^{2N-1}$. Note that inflation properties still hold for $k$ sequences, as the proof of Lemma \ref{le:infl} works in exactly the same way, i.e.~inflating all the sequences increases their LCIS by a factor of $2$.

To obtain $B_i$, we first inflate $A_i$, and then append a tail gadget after
each block $\alpha_i^j$. However, tail gadgets are now more involved.

Let $s$ denote the largest element appearing in $A_1, A_2, \ldots, A_k$.
Then the blocks of $B_i$ are
\[\beta_i^{2j} = \infl(\alpha_i^j) \conc T_i^0, \quad\quad \beta_i^{2j+1} = T_i^1,\]
where $T_i^0$ is the sorted sequence of numbers of the form $2s+x$ for
$x\in\set{1,\ldots,2^k-1}$ such that the $i$-th bit in the binary representation
of $x$ equals $0$, while $T_i^1$ contains those with $i$-th bit set to $1$. Note
that for $k=2$ this exactly leads to the construction from Section~\ref{sec:lcis}.

During one construction step, every block doubles its size, and constant number of elements (precisely, $2^k-1$) is added for every original block. Therefore, the length $L(N)$ of $N$-block sequences satisfies the recursive equation:
\[L(2N) = 2 \cdot L(N) + (2^k-1) \cdot N\]
which yields $L(N) = \Oh{N \log N}$. Note also that the size of the alphabet $S(N)$ used in $N$-block sequences gives the equation $S(2N) = 2 S(N) + 2^k - 1$, as a constant number of elements is added in every step. Therefore $S(N) = \Oh{N}$.

\begin{lemma}
The constructed sequences satisfy, for any $j_1, j_2, \ldots, j_k \in \set{0, 1, \ldots, 2N-1}$,
\begin{equation*}
\lcis(\beta_1^0\ldots\beta_1^{j_1},
      \beta_2^0\ldots\beta_2^{j_2},
      \ldots,
      \beta_k^0\ldots\beta_k^{j_k}) = j_1 + j_2 + \cdots + j_k + 2N, 
\end{equation*}
\end{lemma}
\begin{proof}
We prove the claim by induction on $j_1 + j_2 + \cdots + j_k$. In fact, to make the induction work, we need to prove a stronger statement that there always exists a corresponding LCIS that ends on an element less than or equal to $2s+x(j_1,\dots,j_k)$, where $x(j_1,\dots, j_k)$ is the integer given by the binary representation $(j_1 \bmod 2, \dots, j_k \bmod 2)$.

By the inflation properties and the observation that $T_1^0, \dots, T_k^0$ have no common elements, we obtain the base case $\lcis(\beta_1^0, \dots, \beta_k^0) = 2\cdot\lcis(\alpha_1^0,\dots,\alpha_k^0) = 2N$, with a corresponding LCIS using only elements bounded by $2s$, as desired.

Let $j_1, j_2, \ldots, j_k$ be indices with $j_1+\cdots+j_k > 0$. Let us first construct a common increasing subsequence of length at least $j_1 + \cdots + j_k + 2N$. If all indices $j_1,\dots, j_k$ are even, then, for every $i\in\set{1,\ldots,k}$, the prefix $\beta_i^0 \dots \beta_i^{j_i}$ contains $\infl(\alpha_i^0 \dots \alpha_i^{j_i/2})$ as a subsequence. Thus we can find, by inflation properties, a common increasing subsequence of length $2 \cdot (j_1/2 + \cdots + j_k/2 + N) = j_1 + \cdots + j_k + 2N$, as desired. Now, let $j_i$ be any odd index, and let $L$ be the LCIS of the prefixes corresponding to $j_1,\dots, j_{i-1}, j_i - 1, j_{i+1}, \dots, j_k$, which ends on an element bounded by $x(j_1,\dots,j_{i-1},0,j_{i+1},\dots,j_k)$, of length $j_1 + \dots + j_k + 2N -1$ (which exists by the induction hypothesis). Then $L \conc x(j_1,\dots, j_{i-1}, 1, j_{i+1},\dots,j_k)$ is an LCIS for the prefixes corresponding to $j_1,\dots,j_k$: Indeed, $2s+x(j_1,\dots,j_k)$ is a common member of $T_1^{j_1 \bmod 2}, \dots, T_k^{j_k \bmod 2}$, the last parts of these prefixes, and this element is larger and appears later in the sequences than all elements in $L$ (since all $T_i^j$'s are sorted in the increasing order).

For the converse, let $L$ denote the LCIS of $\beta_1^0\ldots\beta_1^{j_1}$, $\beta_2^0\ldots\beta_2^{j_2}$, $\ldots$, $\beta_k^0\ldots\beta_k^{j_k}$.
Note that if the last symbol of $L$ does not come from the last blocks,
i.e.~$\beta_1^{j_1}, \beta_2^{j_2}, \ldots, \beta_k^{j_k}$, then $L$ is an LCIS of prefixes corresponding to some $j_1', \dots, j_k'$ with $j_1' + \cdots + j_k' < j_1 + \cdots + j_k$ and the claim follows from the induction hypotheses. Thus, we may assume that $L$ ends on a common symbol of the last blocks.

If all the indices are even, the last blocks share only
elements less than or equal to $2s$ (since $T_1^0, \dots, T_k^0$ share no elements), thus $L$ is the LCIS of $\infl(\alpha_i^0, \dots, \alpha_i^{j_i/2}), i\in \{1,\dots,k\}$ and the claim follows from the inflation properties. Otherwise, the only element the last blocks have in
common is $x(j_1, j_2, \ldots, j_k)$, and thus $L= L' \conc x(j_1,\dots, j_k)$, where $L'$ is the LCIS of prefixes corresponding to some $j_1', \dots, j_k'$ with $j_1' + \cdots + j_k' < j_1 + \cdots + j_k$. Thus, $|L| \le j_1' + \cdots + j_k' + 2N + 1\le j_1 + \cdots + j_k + 2N$, as desired.
\end{proof}

\subsection{Generalizing vector gadgets}
\label{sec:klcis-vg}

Each vector gadget is the concatenation of \emph{coordinate gadgets}.
Coordinate gadgets for $j$-th coordinate use elements from the range
$\{kj + 1, \ldots, kj + k\}$. If a coordinate is $0$, the corresponding gadget
contains all $k$ elements sorted in decreasing order, otherwise the gadget for
the $i$-th sequence skips the $kj+i$ element. Formally,
\[\VG_i(u) = \CG_i^0(u[0]) \conc \CG_i^1(u[1]) \conc \cdots \conc \CG_i^{d-1}(u[d-1]),\]
where
\begin{align*}
\CG_i^j(0) &= \seq{kj + k, kj + (k-1), \ldots, kj + 1}, \\
\CG_i^j(1) &= \seq{kj + k, kj + (k-1), \ldots, kj + (i+1), kj + (i-1), \ldots, kj + 1}.
\end{align*}
Thus, if all $k$ vectors have the $j$-th coordinate equal $1$, there is no
common element in the corresponding gadgets. Otherwise, if at least one, say $i$-th,
vector has the $j$-th coordinate equal $0$, the element $kj+i$ appears in all
coordinate gadgets. Since the coordinate gadgets are sorted in decreasing order,
their LCIS cannot exceed $1$. Therefore,
\[\lcis(\CG_1^j(u_1), \CG_2^j(u_2), \ldots, \CG_k^j(u_k)) = 1 - \prod_{i=1}^k u_i[j],\]
and ultimately
\[\lcis(\VG_1(u_1), \VG_2(u_2), \ldots, \VG_k(u_k)) = d - \sum_{j=0}^{d-1}\prod_{i=1}^k u_i[j].\]

\subsection{Putting pieces together}

We can finally prove our lower bound for $k$-LCIS, i.e., Theorem~\ref{thm:klcis}.

\begin{proof}[Proof of Theorem~\ref{thm:klcis}]
Let $\mathcal{U}_1,\dots,\mathcal{U}_k \subseteq \{0,1\}^d$ be a $k$-OV instance with $|\mathcal{U}_i| = n$.
By at most doubling the number of vectors in each set, we may assume without loss of generality that $n$ is a power of two.

We construct separator sequences consisting of $n$ blocks.
Inflate the sequences $\ceil{\log_2 kd}$ times, thus increasing their
length by a factor $\ell = 2^{\ceil{\log_2 kd}}$, and subtract from all their
elements a constant large enough for them to become smaller than all elements
of vector gadgets. Let $A_i = \alpha_i^0 \dots \alpha_i^{n-1}$ denote the thus constructed separator sequence corresponding to set $\mathcal{U}_i$. 

Analogously (and as in the proof of Theorem~\ref{thm:lcis}), construct, for each $i\in \{1,\dots,k\}$, the separator sequence $\hat A_i = \hat\alpha_i^0, \dots, \hat\alpha_i^{n-1}$ by reversing $A_i$, replacing each element by its additive inverse, and adding a constant large enough to make all the elements larger than vector gadgets (note that each $\hat\alpha_i^j$ equals the reverse of $\alpha^{n-j-1}_i$, with negated elements, shifted by an additive constant). In this way, the analogous property to Equation~\eqref{eq:k-sep} holds for suffixes instead of prefixes.

Finally, construct sequences $X_1, X_2,\ldots,X_k$ by defining
\[ X_i = \alpha_i^0 \VG_i(u_i^0) \hat\alpha_i^0
         \alpha_i^1 \VG_i(u_i^1) \hat\alpha_i^1
         \ldots
         \alpha_i^{n-1} \VG_i(u_i^{n-1}) \hat\alpha_i^{n-1}, \]
where the $\VG_i$ are defined as in Section~\ref{sec:klcis-vg}.
It is straightforward to rework the proof of Theorem~\ref{thm:lcis} to verify that these sequences fulfill
\[\lcis(X_1, X_2, \ldots, X_k) = \ell \cdot (k\cdot(n - 1) + 2n) + d - m,\]
where
$m=\min_{u_1\in\mathcal{U}_1, u_2\in\mathcal{U}_2, \ldots, u_k\in\mathcal{U}_k}
\sum_{j=0}^{d-1}\prod_{i=1}^k u_i[j]$.

By this reduction, an $\Oh{n^{k-\varepsilon}}$ time algorithm for $k$-LCIS would yield an $\Oh{n^{k-\varepsilon'}}$ time $k$-OV algorithm (for any dimension $d=n^{o(1)}$), thus refuting Hypothesis~\ref{hyp:kov} and, in particular, SETH.
\end{proof}

\section{Hardness of $k$-LCWIS}

We shortly discuss the proof of Theorem~\ref{thm:lcwis}.

\begin{proof}[Proof sketch of Theorem~\ref{thm:lcwis}]

Note that our lower bound for $k$-LCIS almost immediately yields a lower bound for $k$-LCWIS: Trivially, each common increasing subsequence of $X_1, \dots, X_k$ is also a common \emph{weakly} increasing subsequence. The claim then follows after carefully verifying that, in the constructed sequences, we cannot obtain longer common weakly increasing subsequences by reusing some symbols. 

Our claim for $k$-LCWIS is slightly stronger, however. In particular, we aim to reduce the size of the alphabet over which all the sequences are defined. For this, the key insight is to replace the \emph{inflation} operation $\infl(\seq{a_0, \dots, a_{n-1}}) = \seq{2a_0-1, 2a_0, \dots, 2a_{n-1} -1, 2a_{n-1}}$ by
\[ \mathrm{inflate}'(\seq{a_0, \dots, a_{n-1}}) = \seq{a_0, a_0, \dots, a_{n-1} , a_{n-1}}, \]
which does not increase the alphabet size, but still satisfies the desired property for $k$-LCWIS.

Replacing this notion in the proof of Theorem~\ref{thm:klcis}, we obtain final sequences $X_1, \dots, X_k$ by combining separator gadgets over alphabets of size $\Oh{\log n}$ with vector gadgets over alphabets of size $\Oh{d}$, where $d$ is the dimension of the vectors in the $k$-OV instance. Correctness of this construction under $k$-LCWIS can be verified by reworking the proof of Theorem~\ref{thm:klcis}. Thus, we construct hard $k$-LCWIS instances over an alphabet of size $\Oh{\log n + d}$, and the claim follows.
\end{proof}

\section{Strengthening the Hardness}

In this section we show that a natural combination of constructions proposed in the previous sections with the idea of \emph{reachability gadgets} introduced by Abboud et al.~\cite{AbboudHVWW16} lets us strengthen our lower bounds to be derived from considerably weaker assumptions than SETH. Before we do this, we first need to introduce the notion of branching programs.

A \emph{branching program} of width $W$ and length $T$ on $N$ Boolean input variables $x_1, x_2, \ldots, x_N \in \{0,1\}$ is a directed acyclic graph on $W \cdot T$ nodes, arranged into $T$ \emph{layers} of size $W$ each. A node in the $k$-th layer may have outgoing edges only to the nodes in the $(k+1)$-th layer, and for every layer there is a variable $x_i$ such that every edge leaving this layer is labeled with a constraint of the from $x_i=0$ or $x_i=1$. There is a single \emph{start} node in the first layer and a single \emph{accept} node in the last layer. We say that the branching program \emph{accepts} an input $x \in \{0,1\}^N$ if there is a path from the start node to the accept node which uses only edges that are labeled with constraints satisfied by the input $x$.

The expressive power of branching programs is best illustrated by the theorem of Barrington~\cite{Barrington89}. It states that any depth-$d$ fan-in-$2$ Boolean circuit can be expressed as a branching program of width $5$ and length $4^d$. In particular, $\mathsf{NC}$-circuits can be expressed as constant width quasipolynomial length branching programs.

Given a branching program $P$ on $N$ input variables, the Branching Program Satisfiability problem (BP-SAT) asks if there exists an assignment $x \in \{0,1\}^N$ such that $P$ accepts $x$. Abboud et al.~\cite{AbboudHVWW16} gave a reduction from BP-SAT to LCS (and some other related problems, such as Edit Distance) on two sequences of length $2^{N/2} \cdot T^{\Oh{\log W}}$. The reduction proves that a strongly subquadratic algorithm for LCS would imply, among others, exponential improvements over exhaustive search for satisfiability problems not only on CNF formulas (i.e.~refuting SETH), but even $\mathsf{NC}$-circuits and circuits representing $\oh{\sqrt{n}}$-space nondeterministic Turing machines. Moreover, even a sufficiently large polylogarithmic improvement would imply nontrivial results in circuit complexity. We refer to the original paper~\cite{AbboudHVWW16} for an in-depth discussion of these consequences.

In this section we prove Theorem~\ref{thm:lcis-bp} and thus show that a subquadratic algorithm for LCIS would have the same consequences. Our reduction from OV to LCIS (presented in Section~\ref{sec:lcis}) is built of two ingredients: (1) relatively straightforward vector gadgets, encoding vector inner product in the language of LCIS, and (2) more involved separator sequences, which let us combine many vector gadgets into a single sequence. In order to obtain a reduction from BP-SAT we will need to replace vector gadgets with more complex reachability gadgets. Fortunately, reachability gadgets for LCIS can be constructed in a similar manner as reachability gadgets for LCS proposed in~\cite{AbboudHVWW16}.

\newcommand{\RGx}[2]{\mathrm{RG}_\x^{#1\to#2}}
\newcommand{\RGy}[2]{\mathrm{RG}_\y^{#1\to#2}}
\newcommand{\bRGx}[2]{\overline{\mathrm{RG}}_\x^{#1\to#2}}
\newcommand{\bRGy}[2]{\overline{\mathrm{RG}}_\y^{#1\to#2}}
\newcommand{\fRGx}{\mathrm{RG}_\x}
\newcommand{\fRGy}{\mathrm{RG}_\y}

\begin{proof}[Proof sketch of Theorem~\ref{thm:lcis-bp}]

Given a branching program, as in~\cite{AbboudHVWW16}, we follow the split-and-list technique of Williams~\cite{Williams05}. Assuming for ease of presentation that $N$ is even, we split the input variables into two halves: $x_1,\ldots,x_{N/2}$ and $x_{N/2+1},\ldots,x_N$. Then, for each possible assignment $a \in \{0,1\}^{N/2}$ of the first half we list a reachability gadget $\fRGx(a)$, and similarly, for each possible assignment $b \in \{0,1\}^{N/2}$ of the second half we list a reachability gadget $\fRGy(b)$. We shall define the gadgets such that there exists a constant $C$ (depending only on the branching program size) such that $\lcis(\fRGx(a),\fRGy(b))=C$ if and only if $a \conc b$ is an assignment accepted by the branching program, and otherwise $\lcis(\fRGx(a),\fRGy(b))<C$. The reduction is finished by applying Lemma~\ref{lem:2-lcis-main} to the constructed gadgets in order to obtain two sequences such that their LCIS lets us determine whether a satisfying assignment to the branching program exists. The rest of the proof is devoted to constructing suitable reachability gadgets.

We assume without loss of generality that $T=2^t+1$ for some integer $t$. For every $k\in\{0,1,\ldots,t\}$ and for every two nodes $u, v$ being $2^k$ layers apart from each other we want to construct two reachability gadgets $\RGx{u}{v}$ and $\RGy{u}{v}$ such that, for some constant $C_k$,
\[\lcis(\RGx{u}{v}(a), \RGy{u}{v}(b)) \begin{cases}
= C_k & \text{if there is a path from }u\text{ to }v\text{ satisfied by } a \conc b, \\
< C_k & \text{otherwise},
\end{cases}\]
holds for all $a,b \in \{0,1\}^{N/2}$.

Consider $k=0$, i.e., designing reachability gadgets for nodes in neighboring layers $L_j$ and $L_{j+1}$. There is a variable $x_i$ such that all edges between $L_j$ and $L_{j+1}$ are labeled with a constraint $x_i=0$ or $x_i=1$. We say \emph{the left half is responsible for $x_i$} if $x_i$ is among the first half $x_1,\dots,x_{N/2}$ of variables; otherwise, we say the \emph{right half is responsible for $x_i$}. We set $\RGx{u}{v}(a)$ to be an empty sequence if the left half is responsible for $x_i$ and there is no edge from $u$ to $v$ labeled $x_i = a_i$; otherwise, we set $\RGx{u}{v}(a) = \seq{0}$. Similarly, $\RGy{u}{v}(b)$ is an empty sequence if the right half is responsible and there is no edge from $u$ to $v$ labeled with $x_i = b_{i-N/2}$; otherwise $\RGy{u}{v}(b) = \seq{0}$. It is easy to verify that such reachability gadgets satisfy the desired property for $C_0=1$.

For $k>0$, let $w_1, w_2, \ldots w_{W}$ be the nodes in the layer exactly halfway between $u$ and $v$. Observe that there exists a path from $u$ to $v$ if and only if there exists a path from $u$ to $w_i$ and from $w_i$ to $v$ for some $i\in\{1,2,\ldots,W\}$.

Let $\bRGx{w_i}{v}$ and $\bRGy{w_i}{v}$ denote the sequences $\RGx{w_i}{v}$ and $\RGy{w_i}{v}$ with every element increased by a constant large enough so that all elements are larger than all elements of $\RGx{u}{w_i}$ and $\RGy{u}{w_i}$. Observe that $\lcis(\RGx{u}{w_i}(a)\conc\bRGx{w_i}{v}(a),\RGy{u}{w_i}(b)\conc\bRGy{w_i}{v}(b))$ equals $2\cdot C_{k-1}$ if there is a path $u \leadsto w_i \leadsto v$ satisfied by $a \conc b$, and otherwise it is less than $2\cdot C_{k-1}$. Now, for every $i$ take a different constant $q_i$ and add it to both $\RGx{u}{w_i}\conc\bRGx{w_i}{v}$ and $\RGy{u}{w_i}\conc\bRGy{w_i}{v}$ so that their alphabets are disjoint, and therefore, for $i \neq j$, $\lcis((\RGx{u}{w_i}(a)\conc\bRGx{w_i}{v}(a))+q_i,(\RGy{u}{w_j}(b)\conc\bRGy{w_j}{v}(b))+q_j) = 0$ (where $+$ denotes element-wise addition). Finally, apply Lemma~\ref{lem:2-lcis-main} to these $W$ pairs of concatenated reachability gadgets (where we choose $\delta$ as the maximum length of these gadget) to obtain two reachability gadgets $\RGx{u}{v}$ and $\RGy{u}{v}$ such that $\lcis(\RGx{u}{v}(a), \RGy{u}{v}(b))$ equals $C+2\cdot C_{k-1}$ (for a constant $C$ resulting from the application of Lemma~\ref{lem:2-lcis-main}) if there exists (for some $i\in\{1,2,\ldots,W\}$) a path $u \leadsto w_i \leadsto v$ satisfied by $a \conc b$, and is strictly smaller otherwise, as desired.

Let $u_{\mathrm{start}}$ and $u_{\mathrm{accept}}$ denote the start node and the accept node of the branching program. Then, $\fRGx=\RGx{u_{\mathrm{start}}}{u_{\mathrm{accept}}}$ and $\fRGy=\RGy{u_{\mathrm{start}}}{u_{\mathrm{accept}}}$ satisfy the property that
\[\lcis(\fRGx(a), \fRGy(b)) \begin{cases}
= C_t & \text{if the branching program accepts } a \conc b, \\
< C_t & \text{otherwise}.
\end{cases}\]

Since $\fRGx(a)$ and $\fRGy(b)$ are constructed in $t$ steps of the inductive construction, and each step increases the length of gadgets by a factor of $\Oh{W\log W}$, their final length can be bounded by $\Oh{(W\log W)^t}$, which is $T^{\Oh{\log W}}$. Combining the reachability gadgets $\fRGx(a)$, $a\in \{0,1\}^{N/2}$ and $\fRGy(b)$, $b\in \{0,1\}^{N/2}$ using Lemma~\ref{lem:2-lcis-main} (where we choose $\delta$ as the maximum length of the reachability gadgets) yields the desired strings $X,Y$ of length $2^{N/2} \cdot N \cdot T^{\Oh{\log W}}$ whose LCIS lets us determine satisfiability of the given branching program, thus finishing the proof.
\end{proof}

Similar techniques can be used to analogously strengthen other lower bounds in our paper.
} %% END of \onlyFULL

\section{Conclusion and Open Problems}
\label{sec:conclusion}

We prove a tight quadratic lower bound for LCIS, ruling out strongly subquadratic-time algorithms under SETH. It remains open whether LCIS admits mildly subquadratic algorithms, such as the Masek-Paterson algorithm for LCS~\cite{MasekP80}. Note, however, that our reduction from BP-SAT gives an evidence that shaving many logarithmic factors is immensely difficult. Finally, we give tight SETH-based lower bounds for $k$-LCIS.

For the related variant LCWIS that considers weakly increasing sequences, strongly subquadratic-time algorithms are ruled out under SETH for slightly superlogarithmic alphabet sizes~(\cite{Polak17} and Theorem~\ref{thm:lcwis}). On the other hand, for binary and ternary alphabets, even linear time algorithms exist~\cite{KutzBKK11,Duraj13}. Can LCWIS be solved in time $\Oh{n^{2-f(|\Sigma|)}}$ for some decreasing function $f$ that yields strongly subquadratic-time algorithms for any constant alphabet size $|\Sigma|$?

Finally, we can compute a $(1+\eps)$-approximation of LCIS in
$\Oh{n^{3/2}\eps^{-1/2}\polylog(n)}$
time by an easy observation
\onlySHORT{(see the appendix in the full version).}
\onlyFULL{(see the appendix).}
Can we improve upon this running time or give a matching conditional lower bound? Note that a positive resolution seems difficult by the reduction in Observation~\ref{obs:lcs2lcis}: Any $n^\alpha$, $\alpha > 0$, improvement over this running time would yield a strongly subcubic $(1+\eps)$-approximation for 3-LCS, which seems hard to achieve, given the difficulty to find strongly subquadratic $(1+\eps)$-approximation algorithms for LCS.

\bibliography{lcis}

\onlyFULL{
\clearpage
\section*{Appendix}

\begin{theorem}[folklore, generalization of~\cite{ZhuWWW16}]
For any $k \ge 2$, LCIS of $k$ sequences of length $n$ can be computed in $\Oh{n^k}$ time.
\end{theorem}

\begin{proof}
Let $X_1, X_2, \ldots X_k$ be the input sequences. Let $X[0:i]$ denote the prefix consisting of first $i$ elements of $X$, with $X[0:0]$ being the empty prefix. Now, for every $i_1, \ldots, i_k \in \{0, 1, \ldots, n\}$ we define $R[i_1, \ldots, i_k]$ to be the length of the LCIS of the prefixes  $X_1[0:i_1], X_2[0:i_1], \ldots, X_k[0:i_k]$ with an additional assumption that this common subsequence must end with the element $X_k[i_k]$, i.e. the last element of the last prefix. Observe that it is enough to compute all $R[i_1, \ldots, i_k]$, with the desired answer being simply $\max_{0 \leq i_k \leq n} R[n, \ldots, n, i_k]$.

The algorithm is based on the fact that $R[i_1, \ldots, i_k]$ satisfies the following two-case recurrence:

\begin{itemize}
   \item   \textbf{Case 1.} If $X_k[i_k] \neq X_s[i_s]$ for some $s < k$, the desired common subsequence ends with $X_k[i_k]$ and thus it cannot contain $X_s[i_s]$, so $R[i_1, \ldots, i_s, \ldots, i_k] = R[i_1, \ldots, i_s - 1, \ldots, i_k]$.
   \item   \textbf{Case 2.} If $X_1[i_1] = \ldots = X_k[i_k]$, let us call this common symbol $\sigma$, and observe that $\sigma$ is the last element of LCIS. Consider the next-to-last element: it must be certainly smaller than $\sigma$, and must appear in the $X_k$ sequence at a position earlier than at $X_k[i_k]$. Therefore $R[i_1, \ldots, i_k] = 1 + \max_{j < i_k, X_k[j] < \sigma} R[i_1 - 1, \ldots, i_{k-1}-1, j]$.
\end{itemize}

To obtain the values of $R$ in $\Oh{n^k}$ time, the algorithm iterates through all possible $i_1, \ldots, i_k$ with the $i_k$ loop being the innermost one. Obviously, $R[i_1, \ldots, i_k] = 0$ if any of the indices is $0$. Before every innermost $i_k$ loop, with fixed $i_1, i_2, \ldots, i_{k-1}$, the algorithm checks whether $X_1[i_1] = \ldots = X_{k-1}[i_{k-1}]$. If so, it sets $\sigma = X_1[i_1] = \ldots = X_{k-1}[i_{k-1}]$, otherwise $\sigma = \nullv$.

If $\sigma \neq \nullv$, for every $1 \leq i \leq n$ let $D[i] = \max_{j < i, X_k[j] < \sigma} R[i_1 - 1, \ldots, i_{k-1} - 1, j]$. Observe that $D[i]$ can be obtained from $D[i - 1]$ and $R[i_1 - 1, \ldots, i_{k-1} - 1, i - 1]$ in constant time. Therefore, before the start of the $i_k$ loop, the algorithm can precompute all the $D[i]$ values in $\Oh{n}$ time, as all the needed $R[i_1 - 1, \ldots, i_{k-1} - 1, i]$ values are already known from earlier iterations. 

Throughout the $i_k$ loop the algorithm checks if $X_k[i_k] = \sigma$, which corresponds to Case 2 above. If so, then $R[i_1, \ldots, i_k] = 1 + \max_{j < i_k, X_k[j] < \sigma} R[i_1 - 1, \ldots, i_{k-1}-1, j] = 1 + D[i_k]$, which is already precomputed. If Case 1 holds, then $R[i_1, \ldots, i_s, \ldots, i_k] = R[i_1, \ldots, i_s - 1, \ldots, i_k]$ for some $s < k$. As the index $s$ is easy to find, and the necessary values in $R$ have been computed earlier, this step also works in constant time (assuming $k$ is fixed).

The above algorithm computes only the length of LCIS. However, it can be easily modified to reconstruct the sequence, using the common dynamic programming techniques (e.g. by storing with every value in $R$ a link to the previous element of LCIS).
\end{proof}

\begin{theorem}
A $(1 + \eps)$-approximation of LCIS of sequences $X,Y$ of length $n$ can be computed in $\Oh{n^{3/2}\eps^{-1/2}\polylog(n)}$ time.
\end{theorem}

\begin{proof}
First, delete all integers occurring more than $2\sqrt{n/\eps}$ times
in total in both of the sequences.
Since there are at most $\sqrt{n\eps}$ such integers, this operation decreases
the length of the LCIS by at most $\sqrt{n\eps}$.
In the resulting instance, there are at most $n^{3/2}\eps^{-1/2}$ matching pairs, i.e., indices $i,j$ with $X[i] = Y[j]$. Thus, the exact LCIS in this instance can be computed in time 
$\Oh{n^{3/2}\eps^{-1/2}\log n\log\log n}$ using an algorithm of Chan et al.~\cite{ChanZFYZ07} running in time $\Oh{M \log L \log \log n + n \log n}$, where $L$ is the length of the LCIS of $X$ and $Y$ and $M$ is the number of matching pairs. 
Now, consider two cases. If the algorithm returns a solution $Z$ longer than $\sqrt{n/\eps}$,
then $Z$ is a $(1+\eps)$-approximation of the LCIS of the original instance, since the LCIS is bounded by $L \le |Z| + \sqrt{n \eps} \le (1+\eps)|Z|$.
In the remaining case, it is guaranteed that $L \le |Z| + \sqrt{n\eps} \le (1+\eps)\sqrt{n/\eps}$. Thus, we may
compute the exact LCIS in $\Oh{n^{3/2}\eps^{-1/2}\log n}$ time using
the algorithm running in $\Oh{nL\log \log n + n\log n}$ time~\cite{KutzBKK11}.
\end{proof}

}  % END of \onlyFULL

\end{document}